\documentclass[letterpaper, 10 pt, conference]{ieeeconf}
\IEEEoverridecommandlockouts
\renewcommand{\baselinestretch}{0.9}
\usepackage{setspace}
\usepackage{balance}
\usepackage{graphics}
\usepackage{graphicx}
\usepackage{epsfig}
\usepackage{times} 
\usepackage{amsmath} 
\usepackage{amssymb}  
\usepackage{cite}
\usepackage{multicol}
\usepackage{multirow}
\usepackage{url}
\usepackage[caption=false,font=footnotesize]{subfig}
\usepackage{algorithm,color}
\usepackage[noend]{algpseudocode}
\usepackage{mathrsfs}
\usepackage{mathtools}
\usepackage{tikz-cd}
\usepackage{tikz}
\usetikzlibrary{arrows.meta, positioning, calc}
\usepackage{pgfplots}
\pgfplotsset{compat=newest}
\usepackage{hyperref}
\usepackage{enumerate}
\usepackage[normalem]{ulem}
\usepackage{xcolor}
\colorlet{ab}{teal}

\newcommand{\ang}[1]{\left\langle #1\right\rangle}
\newcommand{\matrices}[1]{\begin{bmatrix} #1\end{bmatrix}}

\newtheorem{assumption}{Assumption}

\newtheorem{remark}{Remark}
\newtheorem{theorem}{Theorem}
\newtheorem{lemma}{Lemma}

\DeclareMathOperator*{\diag}{diag}

\newcommand{\R}{\mathbb{R}}
\newcommand{\C}{\mathbb{C}}

\newcommand{\rank}{\mathrm{rank}\, }
\newcommand{\kernel}{\mathrm{Ker}\, }
\newcommand{\image}{\mathrm{Im}\, }
\newcommand{\sign}{\mathrm{sign}}
\newcommand{\B}{\mathscr{B}}
\newcommand{\X}{\mathscr{X}}
\newcommand{\Y}{\mathscr{Y}}
\newcommand{\U}{\mathscr{U}}
\newcommand{\V}{\mathscr{V}}

\newcommand{\Ss}{\mathscr{S}}
\newcommand{\Rs}{\mathscr{R}}

\newcommand{\W}{\mathscr{W}}
\newcommand{\M}{\mathscr{M}}
\newcommand{\K}{\mathscr{K}}

\newcommand{\Lap}{\mathcal{L}}
\newcommand{\col}{\mathrm{col}}

\usepackage{fancyhdr}

\fancypagestyle{myheader}{
    \fancyhf{} 
    \fancyhead[L]{ } 
    \fancyhead[C]{\footnotesize This work has been submitted to the European Control Conference 2025. Copyright may be transferred without notice, after which this version might no longer be accessible.} 
    \fancyhead[R]{ } 
    \setlength{\headheight}{0pt} 
    \setlength{\headsep}{25pt}    
}

\title{Distributed Unknown Input Observer Design with Relaxed Conditions: Theory and Application to Vehicle Platooning}

\author{Ruixuan~Zhao, Guitao~Yang, Thomas~Parisini, and Boli~Chen
	\thanks{R. Zhao and B. Chen are with the Department of Electronic and Electrical Engineering, University College London, London, UK \tt\small(ruixuan.zhao.22@ucl.ac.uk; boli.chen@ucl.ac.uk).}%
	\thanks{G. Yang and T. Parisini are with the Department of Electrical and Electronic Engineering, Imperial College London, London, UK, and with the KIOS Research and Innovation Center of Excellence University of Cyprus, Cyprus. T. Parinisi is also with the Department of Electronic Systems, Aalborg University, Denmark, and with the Department of Engineering and Architecture, University of Trieste, Italy \tt\small(guitao.yang@imperial.ac.uk; t.parisini@imperial.ac.uk).}%
}

\begin{document}
\pagestyle{myheader}
\begin{titlepage}
    \centering
    \vspace*{\fill}
    {\Huge \bfseries Copyright Statement \\[0.5cm]}

    {\large This work has been submitted to the European Control Conference 2025. Copyright may be transferred without notice, after which this version might no longer be accessible. \\[2cm]}
    \vspace*{\fill}
\end{titlepage}

\clearpage

\maketitle

\begin{abstract}
Designing observers for linear systems with both known and unknown inputs is an important problem in several research contexts, for example, fault diagnosis and fault-tolerant control, and cyber-secure control systems, and presents significant challenges in distributed state estimation due to the limited sensing capabilities of individual nodes. Existing methods typically impose an individual input-to-output rank condition on each estimator node, which severely restricts applicability in practical applications. This paper presents a novel distributed unknown-input observer design scheme based on a geometric approach under much weaker assumptions than the ones available in the literature. By leveraging the properties of the $(C, A)$-invariant (conditioned invariant) subspace at each node, our methodology aims at reconstructing portions of the system state that remain unaffected by local unknown inputs, while integrating these estimates via a network-based information exchange. A case study on vehicle platoon control shows the effectiveness of the proposed approach.
\end{abstract}

\smallskip

\section{Introduction}

Motivated by previous developments in centralized estimation, distributed observers over wireless sensor networks for large-scale monitoring are attracting growing research attention, spanning both continuous-time \cite{kim2019completely,yang2023plug,zhang2022decentralized,wang2017distributed,zhao:ecc24} and discrete-time systems \cite{park2016design,mitra2018distributed,wang2019distributed}.
However, these works assume either no input signals or that input information is accessible to all nodes, a limitation in real-world applications. In many practical large-scale or geographically dispersed systems, individual nodes typically only have access to partial input information (e.g., local inputs), highlighting the need for developing less restrictive methodologies for the design of distributed unknown input observers (DUIO).

Consider a continuous-time linear time-invariant (LTI) system
\begin{equation}\label{eq:con_sys}
    \dot x (t) = Ax(t) + Bu(t),
\end{equation}
where ${x \in \mathbb{R}^n}$ is the state vector and ${u\in \mathbb{R}^m}$ is the control input. ${A\in \mathbb{R}^{n\times n}}$ and ${B\in \mathbb{R}^{n\times m}}$ are known state and input matrices, respectively.
The outputs of the system \eqref{eq:con_sys} are measured by a group of sensors distributed over $N$ nodes, defined as
\begin{equation}\label{eq:output}
    y_i(t) = C_ix(t) \, ,
\end{equation}
where $y_i \in \mathbb{R}^{p_i}$ is the output measurement at node~$i$, $C_i \in \mathbb{R}^{p_i\times n}$ and $i\in\mathbf{N}=\{1,2,\cdots,N\}$.

To better distinguish the locally available inputs and unavailable inputs, we partition the system's inputs $u$ into two components: $u_i\in \mathbb{R}^{l_i}$ and $\bar u_i \in \mathbb{R}^{m - l_i}$ ($l_i \leq p_i \leq n$), where $u_i$ is the locally known input signals and $\bar u_i$ is the locally unknown input signals for Node $i$, respectively. Consequently, we can properly partition the input matrix $B$ and input signal $u$ in the following form 
\begin{equation}\label{eq:inputdecomp}
    Bu(t) = B_i u_i(t) +\bar{B}_i \bar{u}_i(t),
\end{equation}
where $B_i \in \mathbb{R}^{n\times l_i}$ and $\bar{B}_i \in \mathbb{R}^{n\times (m - l_i)}$ are the known and unknown input channels, respectively.
According to \eqref{eq:inputdecomp}, \eqref{eq:con_sys} can be reformulated for each node \( i \in \mathbf{N} \) as follows:
\begin{equation}\label{eq:systemDecomposition}
        \dot x(t) = Ax(t) + B_iu_i(t) + \bar B_i \bar{u}_i(t).
\end{equation}
Inspired by unknown input observers in centralized systems \cite{chen1996design,darouach1994full}, recent studies have explored DUIO design for continuous-time systems \cite{yang2022state,cao2023distributed,cao2023distributedauto,zhu2024distributed}.
The common assumption in the existing results is
\begin{equation}
    \rank (C_i\bar{B}_i)=\rank (\bar{B}_i),\label{eq:rank_con}
\end{equation}
which is a highly restrictive condition imposed on each individual estimator node. For many practical systems, such as the vehicle platoon system discussed in Section~\ref{sec:simu}, this rank condition \eqref{eq:rank_con} does not hold, thereby motivating the approach developed in this paper.
The main contributions can be summarized as follows:
\begin{itemize}
    \item  A geometric approach is proposed to determine the infimal $(C_i, A)$-invariant subspace at each node, ensuring that unknown inputs are filtered out in the corresponding quotient space and that the stability of the induced map is maintained.   
\item A joint geometric condition using the infimal subspace is proposed as a less restrictive alternative to the individual input-to-output rank condition \eqref{eq:rank_con}. This joint condition supports the design of a novel DUIO scheme, where unknown input signals are collaboratively addressed.
\item The effectiveness of the proposed DUIO is demonstrated through a case study on vehicle platooning, underscoring its practical applicability and advantages over existing DUIO schemes.
\end{itemize}

The paper is organized as follows: Section~\ref{sec:Prel} provides the preliminaries 
for the proposed DUIO design.
Section~\ref{sec:geo_approach} presents the geometric approach for identifying the \((C_i, A)\)-invariant subspace. The design of the DUIO, building on our geometric approach, is detailed in Section~\ref{sec:duio}. In the subsequent section, a case study on a vehicle platoon is provided. Finally, concluding remarks are given in Section~\ref{sec:con}.

\section{Preliminaries and Problem Statement}\label{sec:Prel}
\subsection{Notation}
Let $\mathbb{R}$ and $\mathbb{C}$ denote the real and complex numbers, respectively. $\R_{>0}$ is the set of positive real numbers. A symmetric partition on $\mathbb{C}$ is denoted by $\C = \C_g \cup \C_b$ with $\C_g \cap \C_b = \varnothing$ where $C_g$ denotes the partial complex plane consisting of ``good" (e.g. stable) eigenvalues and $C_b$ denotes the partial complex plane consisting of ``bad" (e.g. unstable) eigenvalues. We will denote by $I_n$ an $n\times n$ identity matrix. $\mathbf{1}_{n\times n}$ and 
$\mathbf{0}_{n\times n}$ are $n\times n$ ones matrix and zeros matrix, respectively. 
Let $\|\cdot\|_1$, $\|\cdot\|_2$ and $\|\cdot\|_\infty$ denote the 1-norm, 2-norm and $\infty$-norm of a vector or a matrix, respectively. 
$\otimes$ stands for the Kronecker product, and $\uplus$ represents the union with any common elements repeated. $\sign(\cdot)$ is the sign function. $\mathrm{col}(M_1,M_2,\ldots,M_n)$ denotes the stacked matrix $[M_1^\top,M_2^\top,\cdots,M_n^\top]^\top$ and ${\rm diag}(M_1,M_2,\ldots,M_n)$ denotes the block diagonal matrix composed of $M$'s. $M^\dagger$ represents the pseudo inverse of $M$. $\kappa(M)$ represents the spectrum of $M$. $\sigma_{\min}(M)$ is the minimum singular value of $M$. 

\subsection{A Few Graph Theory Definitions and Results}
The communication network considered in this paper is modeled as an undirected graph \(\mathcal{G} = (\mathbf{N}, \mathcal{E}, \mathcal{A})\), where \(\mathbf{N} = \{1, 2, \dots, N\}\) represents the set of sensor nodes, and \(\mathcal{E} \subseteq \mathbf{N} \times \mathbf{N}\) denotes the set of communication links between the nodes. The adjacency matrix \(\mathcal{A} = [a_{ij}] \in \mathbb{R}^{N \times N}\) characterizes the network, where \(a_{ij} = a_{ji} = 1\) if \((i, j) \in \mathcal{E}\) (indicating a communication link between nodes \(i\) and \(j\)), and \(a_{ij} = a_{ji} = 0\) otherwise.
The Laplacian matrix of \(\mathcal{G}\), denoted as \(\mathcal{L} = [l_{ij}] \in \mathbb{R}^{N \times N}\), is defined by $l_{ij}=\sum_{j=1, j \neq i}^{N} a_{ij}$ if $i=j$, and $l_{ij}=-a_{ij}$ if $i\neq j$.

\begin{lemma}\cite{ren2007information}\label{lemma:eig_Laplacian_1}
Given an undirected, connected graph \(\mathcal{G} = (\mathbf{N}, \mathcal{E}, \mathcal{A})\), \(\frac{\mathbf{1}_{N \times 1}}{\sqrt{N}}\) is both the left and right eigenvector of the Laplacian matrix \(\mathcal{L}\), corresponding to the zero eigenvalue.
\end{lemma}
\subsection{Problem Statement}
Consider the system described by \eqref{eq:systemDecomposition} and \eqref{eq:output}. The objective is to design a class of distributed unknown input observers (DUIOs), denoted as \( \mathcal{O} = \{ \mathcal{O}_i \}_{i \in \mathbf{N}} \), capable of reconstructing the entire system state \( x \). Each observer \( \mathcal{O}_i \) utilizes its locally available input \( u_i \) and output \( y_i \), while accounting for the presence of unknown inputs \( \bar{u}_i \) at each node. Additionally, each node exchanges its local state estimates with neighboring nodes via an undirected communication graph \( \mathcal{G} \). Let $\hat{x}_i$ be the estimated state of node $i$, which is expected to converge to the system state as
\begin{equation*}
\lim_{t \rightarrow \infty} \left\|\hat{x}_i(t)  - x(t) \right\|=0,\quad \forall  i \in \mathbf{N}.
\end{equation*}
\begin{figure}[htp]
    \centering
    \scalebox{0.80}{\begin{tikzpicture}
\def\off{18}
\def\N{5}
\def\R{2.5}
\pgfmathparse{360/\N}
\edef\step{\pgfmathresult}

\colorlet{net}{teal}
\tikzset{comm/.style = {color=net, very thick, dash pattern=on 4pt off 1.5pt}}

\node [circle, 
        draw, 
        color=net, 
        fill=white, 
        text=black, 
        very thick,
        inner sep=6pt] (O1) at (2.5,1.5) {$\mathcal O_{1}$};
\node [circle, 
        draw, 
        color=net, 
        fill=white, 
        text=black, 
        very thick,
        inner sep=6pt] (O2) at (4,0.5) {$\mathcal O_{2}$};
\node [circle, 
        draw, 
        color=net, 
        fill=white, 
        text=black, 
        very thick,
        inner sep=6pt] (O3) at (2, -0.5) {$\mathcal O_{3}$};
\node [circle, 
        draw, 
        color=net, 
        fill=white, 
        text=black, 
        very thick,
        inner sep=6pt] (O4) at (3.5, -1.5) {$\mathcal O_{4}$};
\draw [comm] (O1) -- (O2);
\draw [comm] (O2) -- (O3);
\draw [comm] (O3) -- (O4);
\draw[semithick, color=olive] (1.2,-2.6) -- (4.8,-2.6) -- (4.8,2.3) -- (1.2,2.3) -- cycle;

\node[above] at (3,-2.6) {DUIO};

\node [draw,
        rectangle,
        fill = lightgray,
        minimum width=1.2cm,
        minimum height=4cm] (sys) at (-1,0) {Plant}; 

\draw [-latex, semithick, color=gray] (-0.4,1.5) -- (O1) node[midway, above, align=left] {};
\draw [-latex, semithick, color=gray] (-0.4,0.5) -- (O2) node[midway, above, align=left] {};
\draw [-latex, semithick, color=gray] (-0.4,-0.5) -- (O3) node[midway, above, align=left] {};
\draw [-latex, semithick, color=gray] (-0.4,-1.5) -- (O4) node[midway, above, align=left] {};
\node[above] at (0.5,1.5) {$u_1$\ $y_1$};
\node[above] at (0.5,0.5) {$u_2$\ $y_2$};
\node[above] at (0.5,-0.5) {$u_3$\ $y_3$};
\node[above] at (0.5,-1.5) {$u_4$\ $y_4$};

\draw [-latex, semithick, color=gray] (O1) -- (6,1.5) node[midway, above, align=left] {};
\draw [-latex, semithick, color=gray] (O2) -- (6,0.5) node[midway, above, align=left] {};
\draw [-latex, semithick, color=gray] (O3) -- (6,-0.5) node[midway, above, align=left] {};
\draw [-latex, semithick, color=gray] (O4) -- (6,-1.5) node[midway, above, align=left] {};
\node[right] at (6,1.5) {$\hat{x}_1$};
\node[right] at (6,0.5) {$\hat{x}_2$};
\node[right] at (6,-0.5) {$\hat{x}_3$};
\node[right] at (6,-1.5) {$\hat{x}_4$};
\end{tikzpicture}}
    \\[-1ex]
    \caption{A distributed unknown input observer consisting of 4 sensor nodes: each local observer $\mathcal{O}_i$ has only access to its local input $u_i$ and local measurement $y_i$. Moreover, each node can share its local estimates with neighboring nodes through a communication network described by the dashed line to collectively reconstruct the whole system state.
    }
    \label{fig:network}
\end{figure}
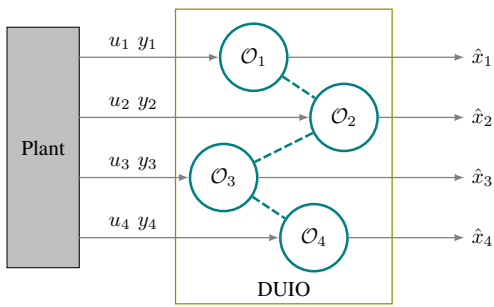
The overall DUIO framework is depicted in Fig.~\ref{fig:network}, illustrating the distributed nature of state estimation and information sharing across the network.

In the remainder of this article, we make the following assumptions:
\begin{assumption}\label{as:connected}
    The communication topology \(\mathcal{G} = (\mathbf{N}, \mathcal{E}, \mathcal{A})\) associated with the network of DUIO is connected.
\end{assumption}
\begin{assumption}\label{as:input_bound}
    The unknown input signal $\bar{u}_i$, $\forall i\in\mathbf{N}$, is bounded such that 
    \begin{equation}
        \|\bar{u}_i\|_\infty \le \bar{u}_{\max}.
    \end{equation}
where $\bar{u}_{\max}$ is a known constant.
\end{assumption}

\section{A Geometric Subspace Decomposition}\label{sec:geo_approach}
\subsection{Element of Geometric Approach}
\subsubsection{Notations and Definitions}
Let $A:\X \rightarrow \X$ be an endomorphism, and let $\W \subseteq \X$ be a subspace with insertion map $W:\W \rightarrow \X$, i.e., $\W=\image W$ and $W$ is monic. A subspace $\W \subseteq\X$ is said to be {\em invariant} with respect to a map $A:\X\rightarrow \X$ if $A\W \subseteq \W$. For an invariant subspace $\W$, we denote by $A|\W: \W \to \W$ the {\em restriction of $A$ to $\W$}, i.e., the unique map satisfying $AW=W(A|\W)$, where $W:\W \rightarrow \X$ is the insertion map of $\W$ in $\X$.
The set $\mathcal{W}=x+\W$ is called a coset of $\W$ in $\X$, and $x\in \X$ is called the representative for $\W$. The set of all cosets of $\W$ in $\X$ is denoted by
$ \X/\W\coloneqq\{x+\W: x\in \X\} $
and is called the quotient (factor) space of $\X$ modulo $\W$. 
Sometimes we denote $\X/\W$ as $\frac{\X}{\W}$ for simplicity. 
We denote by $A|\X/\W$ the map {\em induced on $\X/\W$ by $A$}, satisfying $\left( A|\X \!/\! \W \right) P=PA$, where $P:\X\to\X/\W$ is the canonical projection on $\X/\W$. Let $C:\X \rightarrow \Y$ be a map. If $\Ss \subseteq \Y$, $C^{-1}\Ss$ denotes the \textit{inverse} image of $\Ss$ under C, which is defined by
    $C^{-1}\Ss \coloneqq \{x:x\in \X \ \& \  Cx \in \Ss\} \subseteq \X$.
Note that $C^{-1}$ is the inverse image function of map $C$, and as such it is regarded as a function from the set of all subspaces of $\Y$ to those of $\X$. If $\mathscr{R}, \mathscr{S} \subseteq \mathscr{X}$, we define the subspace $\mathscr{R} + \mathscr{S} \subseteq \mathscr{X}$ as $\mathscr{R} + \mathscr{S}= \{r+s: r\in \mathscr{R}\ \&\  s\in \mathscr{S}\}$, and we define the subspace $\mathscr{R} \cap \mathscr{S} \subseteq \mathscr{X}$ as $\mathscr{R} \cap \mathscr{S}= \{x: x\in \mathscr{R}\ \&\  x\in \mathscr{S}\}$. 
The symbol $\oplus$ indicates that the subspaces being added are independent.
We indicate that two vector spaces $\mathscr V$ and $\mathscr W$ are isomorphic by $\mathscr V \simeq \mathscr W$.
Furthermore, $\mathrm{Mat}(A)$ is the matrix representation of map $A$ relative to the given basis pair.
    \subsubsection{$(C,A)$-invariant Subspace}
Let $A:\X\rightarrow \X$ and $C:\X\rightarrow \Y$. We say a subspace $\W \subseteq \X$ is $(C,A)$-invariant if there exists a map $L:\Y\rightarrow \X$ such that
\begin{equation}\label{eq:def(C,A)-inv}
    (A+LC)\W \subseteq \W.
\end{equation}
The class of $L$ for which \eqref{eq:def(C,A)-inv} holds is denoted by $\mathbf{L}(\W)$ if $A$ and $C$ are clear in the context. The notation $L\in \mathbf{L}(\W)$ reads ``$L$ is a friend of $\W$". 

Let $\B\subseteq \X$. We write $\underline{\W}(C,A;\B)$ to denote the class of $(C,A)$-invariant subspaces containing $\B$. $\underline{\W}(C,A;\B)$ is closed under the operation of subspace intersection, hence $\underline{\W}(C,A;\B)$ has an infimal element (i.e., a subspace with minimal dimension) which is denoted by $\W^*(C,A;\B)$ or simply $\W^*(\B)$ if $A$ and $C$ are clear in the context.



\subsubsection{$(A,B)$-invariant Subspace}
Let $A:\X\rightarrow \X$ and $B:\U\rightarrow \X$. We say a subspace $\V \subseteq \X$ is $(A,B)$-invariant if there exists a map $F:\X\rightarrow \U$ such that
\begin{equation}\label{eq:def(A,B)-inv}
    (A+BF)\V \subseteq \V.
\end{equation}
The class of $F$ for which \eqref{eq:def(A,B)-inv} holds is denoted by $\mathbf{F}(\V)$ if $A$ and $B$ are clear in the context. The notation $L\in \mathbf{F}(\V)$ also reads ``$F$ is a friend of $\V$". 

Let $\K \subseteq \X$. We write $\underline{\V}(\K)$ to denote the class of $(A,B)$-invariant subspaces \textit{contained in} $\K$. $\underline{\V}(\K)$ is closed under the operation of subspace addition, hence $\underline{\V}(\K)$ has a supremal element (i.e., a subspace with maximal dimension) which is denoted by $\V^*(A,B;\K)$ or simply $\V^*(\K)$ if $A$ and $B$ are clear in the context.

\subsubsection{Controllability Subspace}
We say a subspace $\Ss \subseteq \X$ is a $(A,B)$ controllability subspace if 
\begin{equation}
    \Rs = \ang{A+BF\,|\, \image BG }
\end{equation}
for some state feedback map $F:\X \rightarrow \U$ and some input mixing map $G:\U \rightarrow \U$. Note that the controllability subspace is a different concept from the controllable subspace\footnote{The controllable subspace can be defined as $\ang{A|\B} := \B + A\B + \cdots + A^{n-1} \B$, where $\B = \image B$}.

Let $\K \subseteq \X$. We denote the class of controllability subspaces \textit{contained in} $\K$ by $\underline{\Rs}(\K)$. Note that $\underline{\Rs}(\K)$ has a supremal element which is denoted by $\Rs^*(A,B;\K)$ (the supremal controllability subspace that is contained in $\K$) or simply $\Rs^*(\K)$ if $A$ and $B$ are clear in the context.

\subsubsection{Insertion Map and Canonical Projection of Dual Space}
Let $\W \subset \X$ and let 
\[\left(\frac{\X}{\W}\right)' = \W^\perp .\]
If $P:\X \rightarrow \X/\W$ is the canonical projection, then $P': \W^\perp \rightarrow \X'$ is the insertion map from $\W^\perp$ to $\X'$. If $W:\W \rightarrow \X$ is the insertion map, then $W':\X'\rightarrow \X'/\W^\perp$ is the canonical projection. The corresponding commutative diagram is depicted in Fig.~\ref{fig:insertion_canonical_dual}. Note that the matrix form of $W$ and $P'$, as well as $P$ and $W'$ can be identical.

\begin{figure}[htp]
\centering
\begin{tikzcd}[column sep=40pt]
\W \arrow{d}[left]{W}  & \W^\perp \arrow{d}[right]{P'} \\
\X \arrow{d}[left]{P}  & \X' \arrow{d}{W'} \\
\X/\W  & \X'/\W^\perp
\end{tikzcd}\\[-1.2ex]
\caption{Insertion map and canonical projection of quotient spaces.}
\label{fig:insertion_canonical_dual}
\end{figure}

\subsection{Subspace Decomposition}\label{sec:space_decomposition}
The problem introduced in the previous section involves distinct unknown input channels \( \bar{B}_i \) at different nodes, as formulated in \eqref{eq:systemDecomposition}. Consequently, it becomes essential to extract the portion of the system's state that remains unaffected by the unknown inputs. To address this challenge without relying on the restrictive assumption in \eqref{eq:rank_con}, we propose a geometric approach that isolates this unaffected state information.

The main idea is to find, at each node, the infimal \((C_i, A)\)-invariant subspace, denoted as \(\W_{g,i}^*\), that contains \(\image \bar{B}_i\). This subspace facilitates assigning the spectrum of the induced map on \(\X/\W_{g,i}^*\), driven by \(A_{L_i}\), where \(A_{L_i} \coloneqq A + L_i C_i\), with \(L_i: \Y_i \to \X\) representing the output injection map for node \(i\). The goal is to configure the spectrum of the restricted map \(A_{L_i}|\X/\W_{g,i}^*\) to reside within the desired region of the complex plane—typically the stable region—through the appropriate selection of \(L_i\).

Thus, we are able to estimate \(P_{\W_{g,i}^*} x\), where \(P_{\W_{g,i}^*}: \X \to \X/\W_{g,i}^*\) is the canonical projection operator. This ensures that the estimate is not influenced by the unknown input \(\bar{u}_i\), while preserving as much state information as possible. In particular, the dimension of \(\kernel P_{\W_{g,i}^*}\) is minimized, ensuring that the projection retains critical information about the system state.
Consider a family of subspaces that satisfy our requirements
\begin{equation}\label{eq:def_W_gfamily}
\begin{split}
    \underline{\W_{g,i}} 
    \coloneqq &
    \{\W_i: \W_i \in \underline{\W_i}(C_i,A;\image \bar{B}_i),\\ 
    &\& \,\exists L_i \in \mathbf{L}(\W_i)\ \mathrm{s.t.}\ \kappa(A_{L_i}|\X/\W_i) \subset \C_g \}.
\end{split}
\end{equation}
Since finding such $\underline{\W_{g,i}}$ is not trivial, we resort to the dual concept of the $(C,A)$-invariant subspace, which is concluded in the following Theorem:
\begin{lemma}\cite[Theorem~6]{massoumnia1986geometric}
\label{lem:dual(C,A)-inv}
    Let $\W \in \X$. $\W$ is $(C,A)$-invariant if and only if $\W^\perp$ is $(A',C')$-invariant.
\end{lemma}
We then define a family of dual subspaces as follows
\begin{equation}\label{eq:def_V_g'family}
\begin{split}
    \underline{\V_{g,i}'} 
    \coloneqq &
    \{\V_i': \V_i' \in \underline{\V_i}(A',C_i';\kernel \bar{B}_i'),\\ 
    &\& \,\exists L_i' \in \mathbf{F}(\V_i')\ \mathrm{s.t.}\ \kappa(A_{L_i}'|\V_i') \subset \C_g \}.
\end{split}
\end{equation}
By virtue of the Lemma~\ref{lem:dual(C,A)-inv}, we know that for every dual space $\V_i' \in \underline{\V_{g,i}}'$, there exists a subspace $\W_i \in \underline{\W_{g,i}}$ such that $\V_i' = \W_i^\top$, i.e., $\V_i' \simeq \X/\W_i$. Moreover, from \cite[Chap. 5.6]{wonham1985linear}, we know that $\underline{\V'_{g,i}}$ has a supreme element $\V_{g,i}'^*$, which can be identified by the following lemma:
\begin{lemma}\cite[Lemma~5.8]{wonham1985linear}
\label{lem:V'_gi*}
    Define $\V_i'^* \coloneq \V^*(A',C_i';\kernel \bar B_i')$ and $\Rs_i'^* \coloneqq \Rs^*(A',C_i';\kernel \bar B_i')$.
    Choose $L_{0,i}' \in \mathbf{F}(\V'^*)$, write $A_{0,i}' = A' + C_i' L_{0,i}'$.
    Let $P_{R_i^*}': \X' \to \X'/\Rs_i'^*$ be the dual space canonical projection, and let $\bar A_{0,i}'$ be the dual map induced in $\X'/\Rs_i'^*$ by $A_{0,i}$.
    Let $\beta_i(\lambda)$ be the minimal polynomial of $\bar A_{0,i}'|(\V'^*/\Rs'^*)$. Factor $\beta_i (\lambda) = \beta_{g,i}(\lambda)\beta_{b,i}(\lambda)$, where the zeros of $\beta_g$ in $\C$ belong to $\C_g$. Define  
    \begin{equation*}
        \bar \X_{g,i}'^* \coloneqq \frac{\V_i'^*}{\Rs_i'^*} \cap \kernel \beta_{g,i} (\bar A_{0,i}')
    \end{equation*}
    Then the subspace $\V_{g,i}'^*$ defined by
    $
        \V_{g,i}'^* \coloneqq P_{R_i^*}'^{-1} \bar{\X}_{g,i}'^*,
    $
    which is the largest member of the family $\underline{\V'_{g,i}}$ defined by \eqref{eq:def_V_g'family}.
\end{lemma}
Lemma~\ref{lem:V'_gi*} identifies the supremal element in the family $\underline{\V'_{g,i}}$, which implies there is an infimal element in $\underline{\W_{g,i}}$, named as $\W_{g,i}^*$. Since $\V_{g,i}'^* = \W_{g,i}^{*\perp}$, a basis of $\V_{g,i}'^*$ can be regarded as a matrix form of the canonical projection $P_{W_{g,i}^*}$, and the basis of $\W_{g,i}^*$ can be calculated easily by knowing $P_{W_{g,i}^*} \W_{g,i}^* =0$.

Once the desired subspace $\W_{g,i}^*$ is identified, we can readily reconstruct $P_{W_{g,i}^*}x$ at each node. This reconstruction remains unaffected by the unknown input $\bar{u}_i$, as illustrated in the commutative diagram in Fig.~\ref{fig:Wg_decompose}. It is worth noting that $\W_{g,i}^*$ serves as the informal $(C_i,A)$-invariant subspace containing $\image \bar B_i$, ensuring the stability of $A_{L_i}|\X/\W_{g,i}^*$. This implies that $P_{W_{g,i}^*}x$ represents the maximum information reconstructable at each node.

\begin{figure}[htp] 
\centering
    \begin{tikzpicture}


\def\xgap{3.25}
\def\ygap{1.5}
\foreach \i in {0,...,3}{
    \foreach \j in {0,...,3}{
        \coordinate (\i\j) at (\i*\xgap, \j*\ygap);
    }
}
\pgfmathatantwo{\ygap}{\xgap}
\def\labangle{\pgfmathresult}

\tikzset{label/.style = {font=\small, midway}}

\node (barU) at (01) {$\bar \U_i$};
\node (XmodWg) at (10) {$\X / \W^*_{g,i}$};
\node (XmodWg2) at (20) {$\X / \W^*_{g,i}$};
\node (X) at (11) {$\X$};
\node (X2) at (21) {$\X$};
\node (Wg) at (12) {$\W^*_{g,i}$};
\node (Wg2) at (22) {$\W^*_{g,i}$};

\draw[-latex, dashed] (barU) -- (XmodWg) node [label, below, rotate=-\labangle] {$0$};
\draw[-latex] (barU) -- (X) node [label, above] {$\bar B_i$};

\draw[-latex] (XmodWg) -- (XmodWg2) node [label, above] {$A_{L_i}|\X/\W_{g,i}^*$};
\draw[-latex] (XmodWg) -- (XmodWg2) node [label, below] {spectrum good};
\draw[-latex] (X) -- (X2) node [label, above] {$A_{L_i}$};
\draw[-latex] (Wg) -- (Wg2) node [label, above] {$A_{L_i}|\W_{g,i}^*$};

\draw[-latex] (Wg) -- (X) node [label, right] {$W_{g,i}^*$};
\draw[-latex] (Wg2) -- (X2) node [label, right] {$W_{g,i}^*$};

\draw[-latex] (X) -- (XmodWg) node [label, right] {$P_{W_{g,i}^*}$};
\draw[-latex] (X2) -- (XmodWg2) node [label, right] {$P_{W_{g,i}^*}$};

\end{tikzpicture}\\[-1.5ex]
    \caption{Commutative diagram of $\W_{g,i}^*$ decomposition at node $i$.}
    \label{fig:Wg_decompose}
\end{figure}
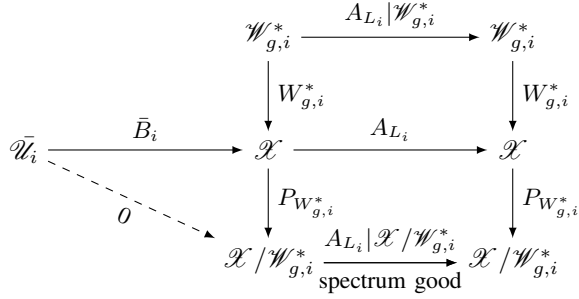

\section{Distributed Unknown Input Observer}\label{sec:duio}
In the DUIO design, we leverage \(\W_{g,i}^*\) to map the shared information (consensus terms) onto the subspace \(\W_{g,i}^*\), while ensuring the convergence of the estimated state within the quotient space \(\X/\W_{g,i}^*\) through the appropriate choice of \(L_i\), as discussed in Section~\ref{sec:space_decomposition}. By imposing a joint geometric condition on the subspace \(\W_{g,i}^*\), the full system state can be locally reconstructed at each sensor node.

The dynamics of the DUIO for node $i$, $i\in\mathbf{N}$, are determined by
\begin{equation}
\begin{split}
    \dot{\hat{x}}_i &= (A+L_iC_i)\hat{x}_i-L_i y_i + \chi W_{g,i}^* {W_{g,i}^*}^\top \sum_{j=1}^N a_{ij} (\hat x_j - \hat x_i)\\
    &\quad \;+ \gamma W_{g,i}^*\mathrm{sign}\left( {W_{g,i}^*}^\top \sum_{j=1}^N a_{ij} (\hat x_j - \hat x_i) \right)  + B_i u_i
\end{split}\label{eq:DUIO_ct}
\end{equation}
where \( L_i \in \mathbb{R}^{n \times p_i} \) and \( W_{g,i}^* \in \mathbb{R}^{n \times w_{g,i}^*} \) have been designed in Section~\ref{sec:space_decomposition}, and \(\chi\) and \(\gamma \in \mathbb{R}_{>0}\) denote the coupling gains for the consensus terms, which will be determined by Theorem~\ref{thm:main_ct}. As presented in \eqref{eq:DUIO_ct}, the first two lines outline the standard consensus-based state estimation update rule, while the third term is specifically designed to compensate for the influence of unknown inputs.

Let $e_i\coloneq x-\hat x_i$ be the estimation error at node $i$. Then, the error dynamics of node $i$ can be expressed as
\begin{equation}
\begin{split}
        \dot{e}_i=&A_{L_i}e_i+\bar{B}_i \bar{u}_i-\chi W_{g,i}^*{W_{g,i}^*}^\top \sum_{j=1}^N a_{ij} (e_i - e_j)\\
        &-\gamma W_{g,i}^*\mathrm{sign}\left( {W_{g,i}^*}^\top \sum_{j=1}^N a_{ij} (e_i - e_j) \right) .
\end{split}
\end{equation}

Let \( T_i \coloneq \begin{bmatrix} W_{g,i}^* & P_{W_{g,i}^*}^\top \end{bmatrix} \in \mathbb{R}^{n \times n} \), where \( W_{g,i}^* \in \mathbb{R}^{n \times w_{g,i}^*} \) and \( P_{W_{g,i}^*} \in \mathbb{R}^{(n - w_{g,i}^*) \times n} \), be an orthonormal transformation adapted to the subspace decomposition at node \( i \). Consequently, the state space can be expressed as
$
    \X = \W_{g,i}^*\oplus\M_i^*, 
$
where \( \M_i^* \simeq \X/\W_{g,i}^* \) represents an arbitrary complementary subspace to \( \W_{g,i}^* \), with its orthonormal basis denoted by \( P_{W_{g,i}^*}^\top \). Therefore, we can obtain
\begin{subequations}
\begin{align*}
   &W_{g,i}^\top P_{W_{g,i}^*}^\top = \mathbf{0}_{w_{g,i}^*\times(n-w_{g,i}^*)},\\
   &{W_{g,i}^*}^\top W_{g,i}^* = I_{w_{g,i}^*},\ 
   P_{W_{g,i}^*} P_{W_{g,i}^*}^\top = I_{n-w_{g,i}^*}.
   \end{align*}
\end{subequations}
%

Referring to the commutative diagram in Fig.~\ref{fig:Wg_decompose}, the following relationships hold:
\begin{subequations}
\begin{align*}
    &W_{g,i}^*(A_{L_i}|\W_{g,i}^*) = A_{L_i}W_{g,i}^*,\\ 
    &(A_{L_i}|\mathscr{X}/\W_{g,i}^*)P_{W_{g,i}^*} = P_{W_{g,i}^*} A_{L_i}. 
\end{align*}
\end{subequations}

Consequently, after applying the transformation by \( T_i \), the matrix representations of \( A_{L_i}|\W_{g,i}^* \) and \( A_{L_i}|\X/\W_{g,i}^* \) are given as follows:
\begin{subequations}\label{eq:AL}
\begin{align}
    &\Tilde{A}_{L_i} \coloneq \mathrm{Mat}(A_{L_i}|\W_{g,i}^*) = W_{g,i}^\top A_{L_i} W_{g,i}, \label{eq:tildeAL}\\
    &\Bar{A}_{L_i} \coloneq \mathrm{Mat}(A_{L_i}|\X/\W_{g,i}^*) = P_{W_{g,i}^*} A_{L_i} P_{W_{g,i}^*}^\top, \label{eq:barAL}
\end{align}
\end{subequations}
where \( \Bar{A}_{L_i} \) is Hurwitz, as the spectrum of \( A_{L_i}|\X/\W_{g,i}^* \), denoted by \( \kappa(A_{L_i}|\X/\W_{g,i}^*) \), is placed within the stable region through the selection of \( L_i \), as described in Section~\ref{sec:space_decomposition}. 

Let $e \coloneq \col(e_1, e_1, \cdots, e_N)$, the dynamics of the entire estimation error can be rewritten as
\begin{multline}
    \dot{e}=\left(A_L  - \chi W_g^* {W_g^*}^\top(\Lap\otimes I_n)\right)e\\
     + \bar{B}\bar{u}- \gamma W_g^* \mathrm{sign}\left({W_g^*}^\top(\Lap\otimes I_n)e\right)
    \label{eq:error_ct}
\end{multline}
where
\begin{subequations}
\begin{align*}
    &A_L = \diag(A_{L_1},\cdots,A_{L_N}),
    &\!\!\!\!\!\!W_g^* = \diag(W_{g,1}^*,\cdots,W_{g,N}^*),\\
    &\bar{B} = \diag(\bar{B}_1,\cdots,\bar{B}_N),\  
    &\!\!\!\!\!\!\bar{u} = \col(\bar{u}_1,\cdots,\bar{u}_N).
    \end{align*}
\end{subequations}
Similarly to \( W_g^* \), define \( M = \mathrm{diag}(P_{W_{g,1}^*}^\top, \dots, P_{W_{g,N}^*}^\top) \), from which we construct an orthonormal transformation matrix \( T \in \mathbb{R}^{nN \times nN} \), where
\begin{equation}\label{eq:basis_T}
    T \coloneqq \begin{bmatrix} W_g^* & M \end{bmatrix}.
\end{equation}
Under this transformation, the matrix representation of the operator \( A_L - \chi W_g^* W_g^{*\top} (\Lap \otimes I_n) \) with respect to the new basis \( T \) is given by:
\[
    T^\top \left( A_L - \chi W_g^* W_g^{*\top} (\Lap \otimes I_n) \right) T = \begin{bmatrix}
        A_a & \star \\
        \mathbf{0} & A_b
    \end{bmatrix},
\]
where
\begin{subequations}
\begin{align*}
    &A_a = \Tilde{A}_L - \chi W_g^{*\top} (\Lap \otimes I_n) W_g^*, \\ 
    &A_b = \mathrm{diag}(\bar{A}_{L_1}, \dots, \bar{A}_{L_N}), \ 
    \Tilde{A}_L = \mathrm{diag}(\Tilde{A}_{L_1}, \dots, \Tilde{A}_{L_N}).
\end{align*}
\end{subequations}
The term \( \star \) is omitted as it does not influence the eigenvalues of the transformed matrix.

Let $\epsilon\coloneq T^\top e$, and $\epsilon\coloneq[\epsilon_a^\top\ \epsilon_b^\top]^\top$, the expression of \eqref{eq:error_ct} in the new coordinate can be reformulated as
\begin{equation}
\begin{split}
        \begin{bmatrix}
        \dot\epsilon_a\\\dot\epsilon_b
    \end{bmatrix}=&\begin{bmatrix}
        A_a&\star\\
        \mathbf{0}&A_b
    \end{bmatrix}\begin{bmatrix}
        \epsilon_a\\\epsilon_b
    \end{bmatrix}+\begin{bmatrix}
        {W_g^*}^\top\bar{B}\\
        \mathbf{0}
    \end{bmatrix}\bar{u}\\
    &-\begin{bmatrix}
        \gamma\mathrm{sign}\left({W_g^*}^\top(\Lap\otimes I_n)(W_g^*\epsilon_a+M\epsilon_b) \right)\\
        \mathbf{0}
    \end{bmatrix}
\end{split}\label{eq:epsilon}
\end{equation}
%
From \eqref{eq:epsilon}, it can be seen that the overall error dynamics of the DUIO decompose into two distinct components: \( \epsilon_a \), which is stabilized by two consensus terms related to the subspace \( \W_{g,i}^* \), and \( \epsilon_b \), which is stabilized by local output injections associated with the quotient space \( \X/\W_{g,i}^* \). 
The following assumption is required to obtain the main results.
\begin{assumption}\label{as:ct}
The subspace $\W_{g,i}^*$ of each node $i$ has the following joint property:
    \begin{equation}
        \bigcap_{i=1}^N\W_{g,i}^* = 0.\label{con:con_for_ct_system}
    \end{equation}
\end{assumption}
%
%
\begin{remark}
    When there are no unknown inputs at all the nodes, i.e., $\bar{B}_i = 0,\,\forall i \in \mathbf{N}$, $\W^*_{g,i}$ reduces to the undetectable subspace at node~$i$. As a consequence, Assumption~\ref{as:ct} becomes identical to the joint detectability assumption imposed in the nominal distributed observer design (see, for example, \cite{yang2023plug} and \cite{wang2017distributed}).  Furthermore, it can be analytically demonstrated that the assumptions underlying existing DUIO designs are also encompassed by Assumption~\ref{as:ct}. A detailed discussion of this will be provided in the journal version, as it lies beyond the scope of the present paper.
\end{remark}
\smallskip
%

%

The following main theorem demonstrates the selection of the coupling gains \(\chi\) and \(\gamma\) that ensure a valid DUIO design.
\begin{theorem}\label{thm:main_ct}
Consider the DUIO designed in \eqref{eq:DUIO_ct} under Assumptions~\ref{as:connected}-\ref{as:input_bound}, for each node $i$, $i\in\mathbf{N}$, if $L_i$ is chosen such that $\kappa(A_{L_i}|\X/\W_{g,i}^*)=\bar{\Lambda}_i\subset\C_g$ and the following conditions hold
\begin{equation}
    \begin{split}
        \chi&> \frac{\left\|\Tilde{A}_L\right\|_2}{\sigma_{\min}\left({W_g^*}^\top(\Lap\otimes I_n)W_g^*\right)},\\ 
        \gamma&>\bar{u}_{\max}\max_{i\in\mathbf{N}}\left(\left\| \bar{B}_i \right\|_1\right)\max_{i\in\mathbf{N}}\left(\left\|W_{g,i}^*\right\|_{\infty}\right)
    \end{split}\label{con:gain}
\end{equation}
Then, the estimation error $e(t)$ converges to $0$ asymptotically.
\end{theorem}
\begin{proof}
Since \( A_b \) is Hurwitz, it follows directly that \( \lim_{t \to \infty} \epsilon_b = 0 \), given that \( \dot{\epsilon}_b = A_b \epsilon_b \) (see \eqref{eq:epsilon}). Thus, it remains to prove the stability of the following subsystem
\begin{equation}\label{eq:epsilon_a}
    \dot\epsilon_a = A_a\epsilon_a + {W_g^*}^\top \bar{B}\bar{u}\!-\!\gamma\,\mathrm{sign}\left({W_g^*}^\top(\Lap\otimes I_n)W_g^*\epsilon_a \right)\,.
\end{equation}
To simplify the notation, let \( \mathbf{Q} \coloneqq {W_g^*}^\top (\Lap \otimes I_n) W_g^* \), which under Assumptions~\ref{as:connected} and \ref{as:ct} is positive definite \cite{kim2019completely}.
Now, consider the following Lyapunov candidate
$
    V(\epsilon_a) = \epsilon_a^\top \mathbf{Q} \epsilon_a.
$
The time derivative of \( V \) along the trajectory of \eqref{eq:epsilon_a} follows
\begin{equation*}\vspace{-5pt}
\begin{split}
        \dot{V}&=2\epsilon_a^\top\left(\mathbf{Q}A_a\right)\epsilon_a + 2\epsilon_a^\top \mathbf{Q} {W_g^*}^\top \bar{B} \bar{u} - 2\gamma \epsilon_a^\top \mathbf{Q} \sign\left(\mathbf{Q}\epsilon_a \right)\\
        &=2(\mathbf{Q}\epsilon_a)^\top\left(\Tilde{A}_L \mathbf{Q}^{-1}-\chi I_{\sum_{i\in\mathbf{N}}w^*_{g,i}}\right)(\mathbf{Q}\epsilon_a)\\
        &\quad + 2\bar{u}^\top\bar{B}^\top {W_g^*}(\mathbf{Q}\epsilon_a) - 2\gamma(\mathbf{Q}\epsilon_a)^\top  \sign\left(\mathbf{Q}\epsilon_a \right)
\end{split}
\end{equation*}
Let \( \zeta \coloneqq \mathbf{Q}\epsilon_a \), 
we further obtain
\begin{equation*}\vspace{-5pt}
    \begin{split}
        \dot{V}&\le 2\left(\left\|\Tilde{A}_L \mathbf{Q}^{-1}\right\|\!-\!\chi\right)\|\zeta\|_2^2 + 2\left\|\bar{u}^\top \bar{B}^\top W_g^* \zeta \right\|_{\infty} \!\!- \!2 \gamma\|\zeta\|_1\\
        &\le 2\left(\left\|\Tilde{A}_L\right\| \left\|\mathbf{Q}^{-1}\right\|-\chi\right)\|\zeta\|_2^2 \\
        &\quad+ 2\left\|\bar{u}^\top \right\|_{\infty} \left\|\bar{B}^\top \right\|_{\infty} \left\| W_g^* \right\|_{\infty} \left\| \zeta\right\|_{\infty} - 2 \gamma\|\zeta\|_1
    \end{split}
\end{equation*}
\begin{equation*}
    \begin{split}
        &\le -2 \left(\chi - \left(\sigma_{\min}\left(\mathbf{Q}\right)\right)^{-1}\left\|\Tilde{A}_L\right\|\right)\|\zeta\|_2^2 \\
        &\quad - 2 \left( \gamma - \bar{u}_{\max}\max_{i\in\mathbf{N}}\left(\left\| \bar{B}_i \right\|_1\right)\max_{i\in\mathbf{N}}\left(\left\|W_{g,i}^*\right\|_{\infty}\right)\right)\|\zeta\|_1
    \end{split}
\end{equation*}

The conditions in \eqref{con:gain} ensure that \(\dot{V}\) is negative definite. Consequently,  \(\epsilon\) will asymptotically converge to zero. Since \(e = T\epsilon\), where \(T\) is defined in \eqref{eq:basis_T}, the convergence of \(e(t)\) is guaranteed, thus completing the proof.
\end{proof}

Given the state estimate $\hat{x}_i$, the unknown input at each sensor node can also be estimated by
\begin{equation}\label{eq:unknown_input}
    \hat{\bar{u}}_i =  \gamma\bar{B}_i^\dagger W_{g,i}^* \sign\left( {W_{g,i}^*}^\top \sum_{j=1}^N a_{ij} (\hat x_j - \hat x_i) \right).
\end{equation}
The state and unknown input estimates are useful in control design, as demonstrated by the case study in Section~\ref{sec:simu}.
\section{Case Study}\label{sec:simu}
Consider a vehicle platoon consisting of four vehicles interconnected via a predecessor-following (PF) information flow topology, as shown in Fig~\ref{fig:platoon}. 
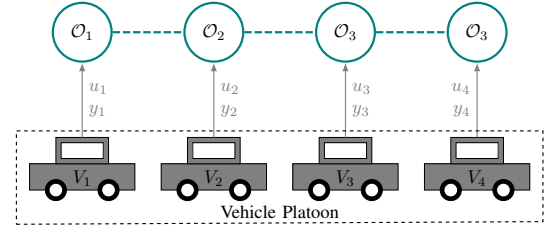
\begin{figure}[htp] 
\centering
\scalebox{0.70}{\begin{tikzpicture}
\colorlet{net}{teal}
\tikzset{comm/.style = {color=net, very thick, dash pattern=on 4pt off 1.5pt}}

    \newcommand{\drawcar}[1]{

        \draw[fill=gray] (#1,0) rectangle (#1+2,0.5);

        \draw[fill=gray] (#1+0.5,0.5) rectangle (#1+1.5,1);

        \draw[fill=black] (#1+0.5,0) circle (0.2); 
        \draw[fill=black] (#1+1.5,0) circle (0.2);

        \draw[fill=white] (#1+0.5,0) circle (0.13); 
        \draw[fill=white] (#1+1.5,0) circle (0.13);

        \draw[fill=white] (#1+0.6,0.6) rectangle (#1+1.4,0.9);
    }

    \drawcar{0}{1}  
    \drawcar{2.5}{2}  
    \drawcar{5}{3} 
    \drawcar{7.5}{4} 

    
    \node [circle, 
        draw, 
        color=net, 
        fill=white, 
        text=black, 
        very thick,
        inner sep=6pt] (N1) at (1,3) {$\mathcal O_{1}$};
    \node [circle, 
        draw, 
        color=net, 
        fill=white, 
        text=black, 
        very thick,
        inner sep=6pt] (N2) at (3.5,3) {$\mathcal O_{2}$};
    \node [circle, 
        draw, 
        color=net, 
        fill=white, 
        text=black, 
        very thick,
        inner sep=6pt] (N3) at (6,3) {$\mathcal O_{3}$};
    \node [circle, 
        draw, 
        color=net, 
        fill=white, 
        text=black, 
        very thick,
        inner sep=6pt] (N4) at (8.5,3) {$\mathcal O_{3}$};

    \draw[comm] (N1) to (N2);    
    \draw[comm] (N2) to (N3);   
    \draw[comm] (N3) to (N4);  

\draw [-latex, semithick,color=gray] (1,1) -- (N1) node[midway, right, align=left] {$u_1$ \\ $y_1$};
\draw [-latex, semithick,color=gray] (3.5,1) -- (N2) node[midway, right, align=left] {$u_2$ \\ $y_2$};
\draw [-latex, semithick,color=gray] (6,1) -- (N3) node[midway, right, align=left] {$u_3$ \\ $y_3$};
\draw [-latex, semithick,color=gray] (8.5,1) -- (N4) node[midway, left, align=right] {$u_4$ \\ $y_4$};

\draw[semithick, dash pattern=on 2pt off 2pt] (-0.25,-0.65) -- (9.75,-0.6) -- (9.75,1.125) -- (-0.25,1.125) -- cycle;

\node[above] at (4.75,-0.65) {Vehicle Platoon};

\node[above] at (1,-0.05){$V_1$};
\node[above] at (3.5,-0.05){$V_2$};
\node[above] at (6,-0.05){$V_3$};
\node[above] at (8.5,-0.05){$V_4$};

\end{tikzpicture}} \\[-1.2ex]
    \caption{A platoon consisting of 4 vehicles following a unidirectional PF information flow topology and proposed DUIO $\{\mathcal{O}_i\}_{i\in\{1,\cdots,4\}}$ scheme.}
    \label{fig:platoon}
\end{figure}
%
Note that the communication graph among vehicles according to Assumption~\ref{as:connected} for the observer design in \eqref{eq:DUIO_ct} is undirected, which is widely adopted in platoon control \cite{qiang2023distributed}.
Unlike most existing studies that assume a constant velocity for the lead vehicle, here, the leader’s velocity is allowed to vary over time, which is essential for the practical deployment of platooning techniques. As such, the longitudinal dynamics of each vehicle can be modeled, as follows\cite{qiang2023distributed}
\begin{equation*}
\begin{aligned}
     &\dot s_i(t) = v_i(t),\ 
     \dot v_i(t) = a_i(t),\\
     &\dot a_i(t) = -\frac{1}{\tau}a_i(t) + \frac{1}{\tau}u_i(t),\,\forall i =\{1,\cdots,4\},
\end{aligned}
\end{equation*}
where $s_i$, $v_i$, and $a_i$ are the position, velocity, and acceleration of vehicle $i$, $u_i$ is the control input, and $\tau$ denotes the inertial lag of longitudinal dynamics of vehicles. 

This platoon system can be effectively controlled using a well-known control law \cite{Ploeg:itsc11}
\begin{multline}\label{eq:control_law0}
    u_i(t)={u}_{i-1}+k_s\left({s}_{i-1}(t)-s_i(t)-d_{i,i-1}(t)\right)\\+k_v\left({v}_{i-1}(t)-v_i(t)\right)+k_a\left({a}_{i-1}(t)-{a}_i(t)\right), 
\end{multline}
where $i\in\{2,3,4\}$, $k_s$, $k_v$ and $k_a$ are control gain parameters, $d_{i,i-1}$ is the desired inter-vehicle distance, which is set to a constant value $20\mathrm{m}$.
However, achieving consensus requires that each follower vehicle be equipped with appropriate sensors to measure headway distance, velocity difference, and acceleration difference relative to the vehicle ahead. Additionally, the control input ${u}_{i-1}$ of vehicle $i-1$ must be shared with the following vehicle $i$, which could raise cybersecurity concerns (e.g., hacked during the transmission). The DUIO proposed in this paper reduces both the number and cost of required sensors and eliminates the need for sharing control inputs within the vehicular network compared to the benchmark method.

Assume each vehicle $i$ only measures its own position $s_i$ and velocity $v_i$, i.e., $y_i=[s_i,\,v_i]^{\top},\,\forall i=\{1,\cdots,4\}$, which can be achieved with minimal effort for most vehicles. Then, the overall system can be formulated in the form of \eqref{eq:systemDecomposition} with
\begin{equation*}
\begin{aligned}
    &A=I_4\otimes \begin{bmatrix}
        0&1&0\\
        0&0&1\\
        0&0&-\frac{1}{\tau}
    \end{bmatrix},
    &B=I_4\otimes\begin{bmatrix}
        0\\
        0\\
        -\frac{1}{\tau}
    \end{bmatrix}\!=\!\matrices{B_1^{\top}\\B_2^{\top}\\B_3^{\top}\\B_4^{\top}}^{\top},
\end{aligned}
\end{equation*}
\begin{equation*}\vspace{-5pt}
\begin{aligned}
    &\bar{B}_1=\matrices{B_2&B_3&B_4},\ \bar{B}_2=\matrices{B_1&B_3&B_4},\\
    &\bar{B}_3=\matrices{B_1&B_2&B_4},\ \bar{B}_4=\matrices{B_1&B_2&B_3},
\end{aligned}
\end{equation*}
\begin{equation*}
    C=I_4\otimes\begin{bmatrix}
                1&0&0\\
                0&1&0    \end{bmatrix}=\matrices{C_1^\top&C_2^\top&C_3^\top&C_4^\top}^\top,
\end{equation*}
where $\tau$ is given by $0.07$. Clearly, this vehicle platoon system does not meet the rank condition \eqref{eq:rank_con} required by the existing DUIO scheme, which highlights the advantages of our proposed DUIO design.

With the DUIO, the state and input estimates can be plugged in and the control law \eqref{eq:control_law0} is changed to
\begin{multline}\label{eq:control_law}
    u_i(t)=\hat{u}_{i-1}+\sum_{j=1}^{i-1}\left(k_s\left(\hat{s}_i^{(j)}(t)-s_i(t)-d_{i,j}(t)\right)\right.\\ \left.+k_v\left(\hat{v}_i^{(j)}(t)-v_i(t)\right)+k_a\left(\hat{a}_i^{(j)}(t)-\hat{a}_i(t)\right)\right), 
\end{multline}
where $i\in\{2,3,4\}$, $\hat{s}_i^{(j)}$, $\hat{v}_i^{(j)}$ and $\hat{a}_i^{(j)}$ are the position, velocity and acceleration of vehicle $j$ estimated by $\mathcal{O}_i$, and $\hat{u}_{i-1}$ is the input estimate of the preceding vehicle. Note that the feedback of each follower is augmented by tracking error with respect to all preceding vehicles rather than the vehicle immediately ahead for enhanced convergence speed \cite{huang2023plug}. 


The control parameters are given as $k_s=3.5$, $k_v=4$ and $k_a=1$. Besides, the initial positions, velocities and accelerations of vehicles are set as $\setlength{\tabcolsep}{2.0pt}x(0)=\left[\begin{tabular}{cccccccccccc}
    150&22&0&120&21&1.1&90&21.5&0.6&60&20&1.3
\end{tabular}\right]^\top$. According to Theorem~\ref{thm:main_ct}, $\chi$ is set to $82.3039$, $\gamma$ is set to $383.1159$, the insertion maps $W_{g,i}^*$ and the output injection maps $L_i$ are provided in the supplementary document\footnote{\url{https://github.com/RuixuanZhaoEEEUCL/ECC2025.git}} due to space limitation.
\begin{figure}[ht]
    \centering
    \includegraphics[width=0.45\textwidth]{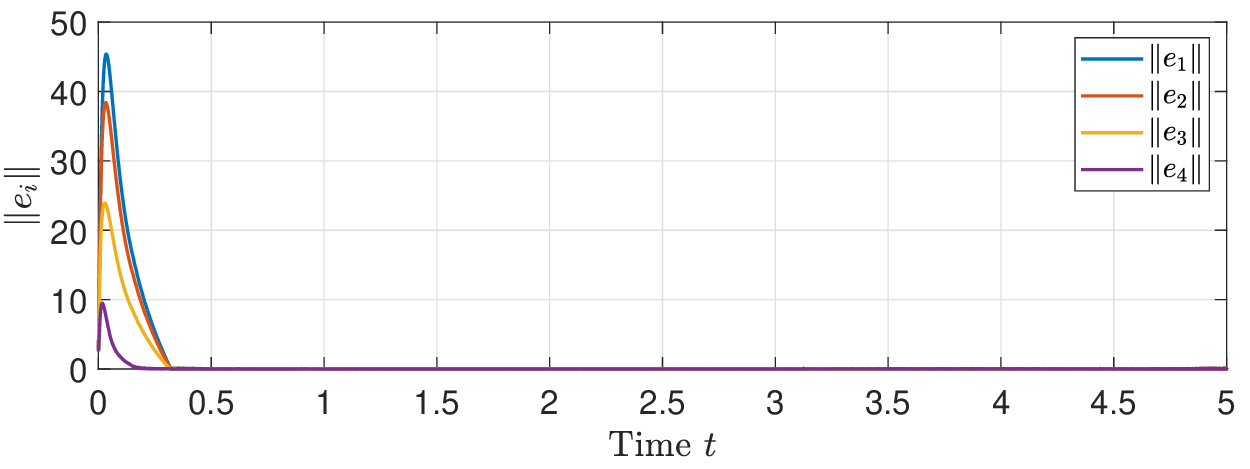}\\[-3.5ex]
     \includegraphics[width=0.45\textwidth]{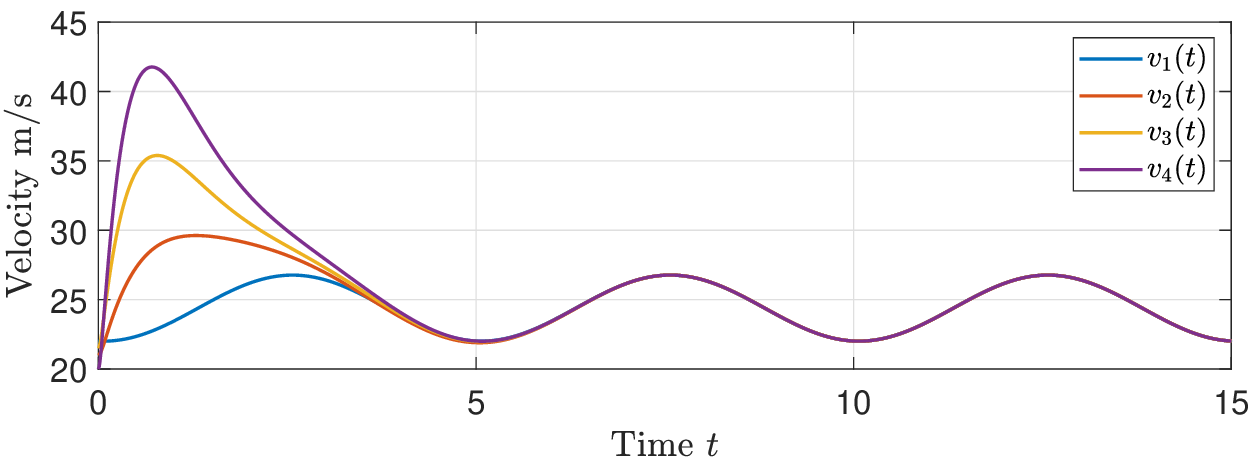}\\[-3.5ex]
      \includegraphics[width=0.45\textwidth]{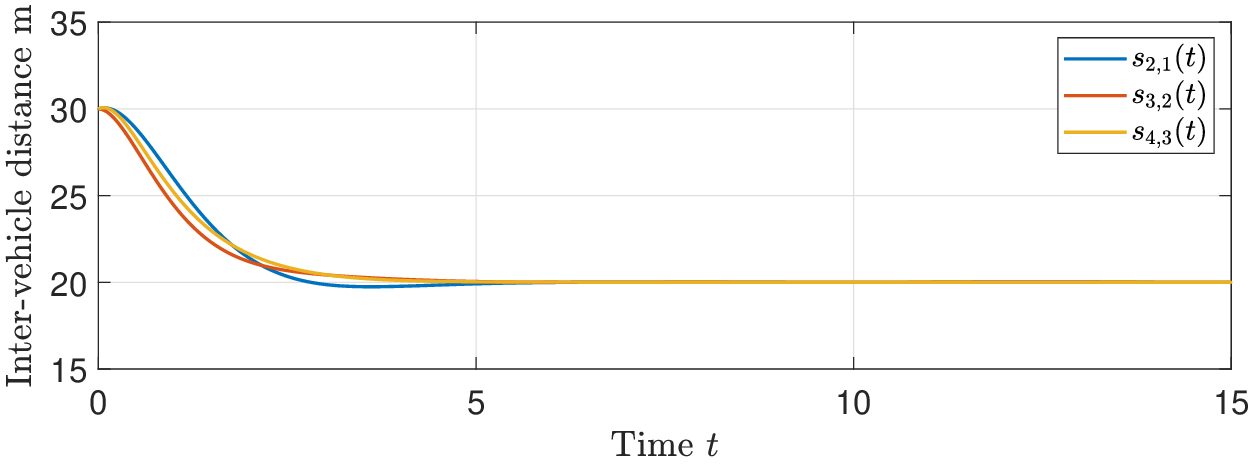}\\[-2ex]
    \caption{System performance subject to a time-varying velocity of the leading vehicle unknown to the follower vehicles. \textbf{Top}: Norm of estimation error at each node. \textbf{Middle}: Velocity trajectories of each vehicle. \textbf{Bottom}: Trajectories of inter-vehicle distance between neighboring vehicles.}
    \label{fig:sim_result_estimation_error}
\end{figure}

The effectiveness of the proposed DUIO design is verified in Fig.~\ref{fig:sim_result_estimation_error}, where the estimation error of each node asymptotically converges to zero within $0.5$ seconds. Velocity and inter-vehicle distance \( s_{i,i-1} = s_i - s_{i-1}, \,\forall i \in \{2,3,4\} \)  profiles are also shown. As can be seen, both can reach the desired values at the steady state. These results validate the capability of the DUIO-based control law in \eqref{eq:control_law} to handle the unknown, time-varying velocity of the leading vehicle, highlighting the potential of the proposed DUIO for practical applications.

\section{Conclusions}\label{sec:con}
In this paper, we present a novel Distributed Unknown Input Observer (DUIO) design for LTI systems using a geometric approach. By leveraging the advantages of $(C, A)$-invariant subspaces and analyzing the invariant zeros at each node, we establish a new joint geometric condition for DUIO, which relaxes the conservative individual rank conditions imposed by existing methods. Simulation results confirm the effectiveness of the proposed DUIO design and demonstrate its practical applicability. In future work, we aim to extend our methodology to discrete-time systems and incorporate process and measurement noise into the framework.


\end{document}